\documentclass[letterpaper, 10 pt, conference]{article}  

\usepackage{times}
\usepackage{t1enc}
\usepackage[export]{adjustbox}
\usepackage{graphicx}
\usepackage{adjustbox }
\usepackage{epstopdf}
\usepackage{subfigure}
\usepackage{latexsym}
\usepackage{amssymb}
\usepackage{amsbsy}
\usepackage{mathtools}
\usepackage{color}
   
\usepackage{amsthm}
\usepackage{amsfonts}
\usepackage{supertabular}

\theoremstyle{plain}
\newtheorem{proposition}{Proposition}

\theoremstyle{definition}
\newtheorem{definition}{Definition}
\newtheorem{remark}{Remark}

\theoremstyle{remark}

\newcommand*\bigcdot{\mathpalette\bigcdot@{.5}}

\newcommand{\beqn}{\begin{eqnarray*}} 
	\newcommand{\eeqn}{\end{eqnarray*}}
\newcommand{\beq}{\begin{eqnarray}} 
\newcommand{\eeq}{\end{eqnarray}}
\def\deff{\stackrel{\triangle}{=}}

\title{\LARGE \bf
On the primal-dual dynamics of Support Vector Machines
}
\author{Krishna Chaitanya Kosaraju, Shravan Mohan and Ramkrishna Pasumarthy\\
\thanks{Krishna Chaitanya and Shravan Mohan are   graduate students in the Department of Electrical  Engineering,
       Indian Institute of Technology Madras, Chennai-600036,
      India. {\tt\small  kkrishnachaitanya89@gmail.com}, {\tt\small  shravan8587@gmail.com }}
\thanks{Ramkrishna is with the Department of Electrical Engineering,  Indian Institute of Technology Madras, Chennai-600036,
India {\tt\small ramkrishna@iitm.ac.in}}}

\begin{document}

\maketitle
\thispagestyle{empty}
\pagestyle{empty}

\begin{abstract}
The aim of this paper is to study the convergence of the
primal-dual dynamics pertaining to Support Vector Machines
(SVM). The optimization routine, used for determining an SVM for classification, is first formulated as a dynamical
system. The dynamical system is constructed such that its
equilibrium point is the solution to the SVM optimization
problem. It is then shown, using passivity theory, that the
dynamical system is global asymptotically stable. In other
words, the dynamical system converges onto the optimal
solution asymptotically, irrespective of the initial
condition. Simulations and computations are provided for
corroboration.
\end{abstract}
\section{INTRODUCTION}
The field of Machine Learning has gained tremendous traction over the past decade with the advent of data compilation from various sectors of the industrial world \cite{michalski2013machine}. The techniques therein have helped the industry gain crucial insights into their processes and make judicious decisions for the future. A ubiquitous component of most Machine Learning algorithms is optimization, where in a suitably chosen cost function is maximized (or minimized) under constraints. In many applications, the cost function and the constraints arise from practical considerations. As far as the optimization routines are concerned, the most well understood class of optimization problems happens to be that of convex optimization  \cite{boyd2004convex}. Convex optimization problems happen to be quite useful and have also percolated into many different application areas. A particularly interesting application is that of classification problems using Support Vector Machines. Support vector machines form a tool set for linear as well as non-linear classification \cite{scholkopf2001learning}. They can also be used effectively for non-linear regression using different kernels. As such, classification itself turns out to be quite useful in the industry; applications range from predicting defaulters in finance sector, predicting claims in the insurance sector and detecting defects in retinopathy \cite{kumar2010data,nguyen1996classification}. 

Gradient-based methods form a fundamental basis of all algorithms for solving convex optimization problems. These gradient algorithms have much to gain from a control and dynamical systems perspective; especially for a better understanding of the underlying system theoretic properties  such as stability, convergence rates, and robustness \cite{khalil1996noninear}. The convergence of gradient-based methods and Lyapunov stability, relate the solution of the optimization problem to the equilibrium point of a dynamical system \cite{Brokett,stegink2017unifying,feijer2010stability,CSS_letters}. In this context, the focus of this paper is continuous time  primal-dual gradient descent method. The formulation has its roots from \cite{kose1956solutions}, where the author constructs a dynamical system whose trajectories converges asymptotically to the solution of a min-max problem (saddle point problem). This framework has two very important characteristics. The first being the equilibrium of the dynamical system is not explicitly known but it is implicitly characterized by the Karush-Kuhn-Tucker (KKT) conditions of an optimization problem (or the optimization problem itself). Secondly, the fact that one can show stability using Lyapunov analysis without the knowledge of equilibrium set or a point. In the literature, such systems are usually called \textit{contracting systems}, a term coined in the seminal paper \cite{lohmiller1998contraction}, where the authors show that the distance between the trajectories contracts exponentially.

\noindent {\em Motivation and contribution:} 
In this paper, we consider a convex optimization formulation of a linear support vector machine problem (usually noted as primal formulation). We next propose the Lagrangian of the constrained optimization problem using which we present its dual formulation. The primal together with its dual forms gives rise to a saddle-point problem. We present the continuous time gradient descent equation for the saddle-point problem, which essentially captures two properties; minimization  of the Lagrangian with respect to the primal variables and maximization of Lagrangian with respective to the dual variables \cite{feijer2010stability,CSS_letters}. Hence these dynamics are usually noted as primal-dual dynamics. We finally present the convergence analysis of these dynamics using tools from passivity \cite{l2gain} and hybrid systems theory \cite{lygeros2003dynamical}. Note that, rewriting the algorithm as dynamical system that converges to the solution of an optimization problem has enabled the use of such systems theory tools for convergence analysis. The main objective of this note is to motivate the dynamical system formulation which will acts as a fundamental entity for future research. Simulation studies are provided to understand the behavior of Lyapunov function and visualizing the results.

The paper is organized as follows. Section II presents a brief overview of convex optimization. Section III presents results on the stability of the continuous time primal-dual dynamics used to solve convex optimization problems. Finally, in Section IV, the ideas are applied to the case of the SVM and simulations are provided for corroboration.
\begin{figure}[t]
	\centering
	\includegraphics[width=1\linewidth]{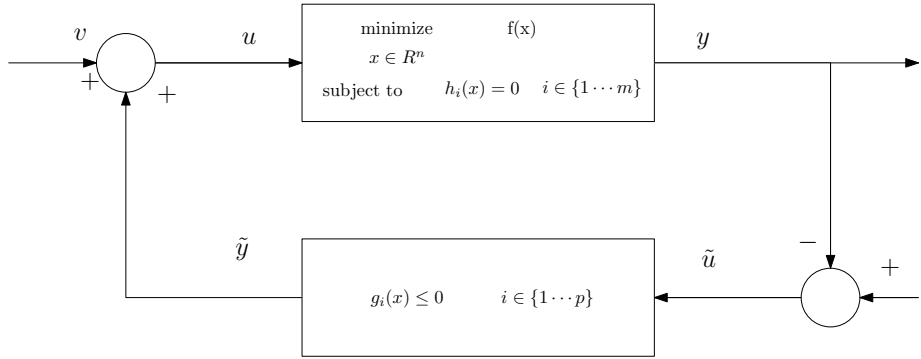}
	\caption{Interconnected optimization}
	\label{fig:my_label}
\end{figure}

\section{Convex optimization}
In this section, we present a brief overview of mathematical tools in convex optimization, that will be useful in the subsequent sections. The standard form of a convex optimization problem contains three parts:
\begin{itemize}
	\item [(i)]	A continuously differentiable convex function $f(x):\mathbb{R}^n\rightarrow \mathbb{R}$ to be minimized over $x$,
	\item[(ii)] 	affine equality contraints $h_{i}(x)=0, \hspace{0.4cm} i = 1,\hdots , m$,
	\item[(iii)] 	continuously differentiable convex inequality constraints of the form $g_{i}(x)\leq 0$, \hspace{0.4cm} $i = 1,\hdots , p$.
\end{itemize}
\noindent This can be written in the following form, commonly known as the primal formulation:
\begin{equation}\label{standard_SOP}
\begin{aligned}
& \underset{x \in \mathbb{R}^{n}}{\text{minimize}}
& & f(x)\\
& \text{subject to}
& & h_{i}(x)=0 \hspace{0.4cm} i = 1,\hdots , m\\
&
& & g_{i}(x)\leq 0 \hspace{0.4cm} i = 1,\hdots , p\\
\end{aligned}
\end{equation}
\begin{figure}[t]
	\centering
 	\includegraphics[width=1\linewidth]{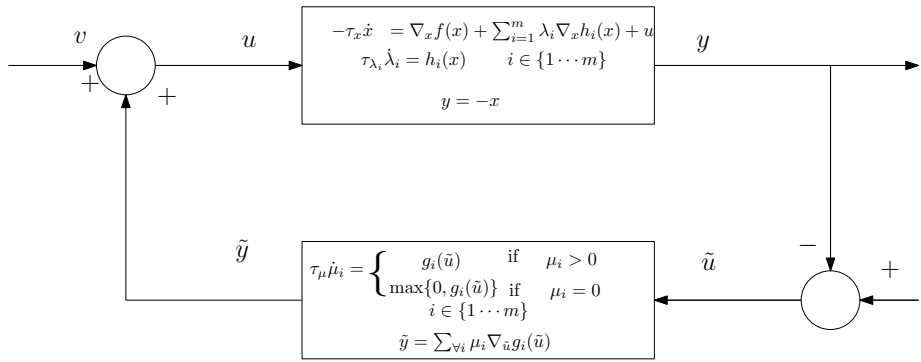}
	\caption{Interconnected primal dual dynamics}
	\label{fig:my_labela}
\end{figure}
\textbf{Karush-Kuhn-Tucker (KKT) conditions}: If the solution $x^\ast$ is optimal to the convex optimization problem \eqref{standard_SOP} then there exist $\lambda_i\in \mathbb{R}$, $i = 1,\hdots , m$ and $\mu_{i}\geq 0$, $i = 1,\hdots , p$ satisfying the following KKT conditions
\begin{eqnarray}\label{standard_KKT}
\nabla_{x} f(x^{*})+\sum_{i=1}^{m}\lambda_{i}\nabla_{x} h_{i}(x^{*})+\sum_{i=1}^{m}\mu_{i}\nabla_{x} g_{i}(x^{*})=0,\nonumber\\
h_{i}(x^{*})=0 \;\;\;\forall i \in \{1,\ldots, m\},\\
g_j(x^\ast)\leq 0,\;\;\mu_j\geq 0,\;\; \mu_jg_j(x^\ast)=0\;\;\;\forall j \in \{1,\ldots, p\}.\nonumber
\end{eqnarray}
\begin{remark}
	Note that the KKT conditions presented above in equation \eqref{standard_KKT} are only necessary conditions. We next present the requirements under which KKT conditions become sufficient.
\end{remark}
\noindent We now define the Lagrangian of the convex optimization \eqref{standard_SOP} as
\begin{eqnarray}\label{convex_Lagrangian}
\mathcal{L}(x,\lambda,\mu)=f(x)+\sum_{i=1}^{m}\lambda_{i} h_{i}(x)+\sum_{i=1}^{m}\mu_{i}g_{i}(x)
\end{eqnarray}
and the Lagrange dual function as
\begin{eqnarray}\label{convex_Lagrange_dual}
L_d(\lambda,\mu)=   
\begin{aligned}
& \underset{x \in \mathbb{R}^{n}}{\text{minimize}}
& & L(x,\lambda,\mu) 
\end{aligned}
\end{eqnarray}
\noindent giving us the following dual problem (corresponding to the primal problem \eqref{standard_SOP})
\begin{equation}\label{standard_SOP_dual}
\begin{aligned}
& \underset{\lambda \in \mathbb{R}^{m},\;\mu\in \mathbb{R}^{p}}{\text{maximize}}
& &  L_d(\lambda,\mu)\\
& \text{subject to}
& & \mu_i\geq 0 \hspace{0.4cm} i = 1,\hdots , p.\\
\end{aligned}
\end{equation}
\begin{remark}
	Dual problem  is always convex, because $L_d$ is always a concave function even when the primal \eqref{standard_SOP} is not convex. If $f^\ast$ and $L_d^\ast$ denotes the optimal values of primal and dual problems respectively, then $L_d^\ast \leq f^\ast$. Therefore dual formulations are used to find the best lower bound of the optimization problem \cite{boyd2004convex,ben2001lectures}. Further, the negative number $L_d^\ast -f^\ast$ denotes the {\em duality gap}. In the case of zero duality gap, we say that the problem \eqref{standard_SOP} {\em satisfies strong duality}.
\end{remark}
\begin{definition}\textbf{Slater's conditions}. The convex optimization problem \eqref{standard_SOP} is said to satisfy Slater's conditions if there exists an $x$ such that 
	$h_{i}(x)=0 \hspace{0.4cm} i = 1,\hdots , m$ and $ g_{i}(x)< 0 \hspace{0.4cm} i = 1,\hdots , p$. This implies that inequality constraints are strictly feasible.
\end{definition}
\begin{remark} If a convex optimization problem \eqref{standard_SOP} satisfies  Slater's conditions then the optimal values of primal and dual problems are equal, that is, \eqref{standard_SOP} satisfies strong duality. Further, in this case the KKT conditions becomes necessary and sufficient.
\end{remark}

\begin{figure}[t]
	\centering
	\includegraphics[width=0.8\linewidth]{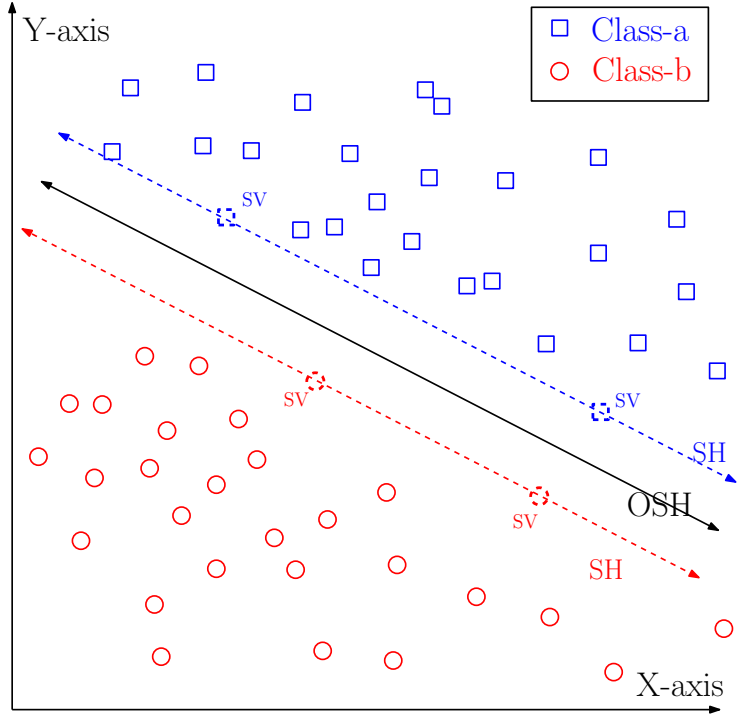}
	\caption{Description of a linear support vector machine.}
	\label{fig::svm_ex}
\end{figure}
\section{ Stability of primal-dual dynamics}
In this section, we present the continuous time primal-dual equations of a convex optimization problem. In \cite{CSS_letters}, we have shown that these dynamics can be described as a feed-back interconnection of two passive dynamical systems. The first being the primal-dual dynamics of an equality constrained optimization problem and the second corresponds to the hybrid dynamics representing the inequality constraint (see Figures \ref{fig:my_label}, \ref{fig:my_labela}). We now briefly revisit these results.

\begin{figure}[t]
	\centering
	\includegraphics[width=0.8\linewidth]{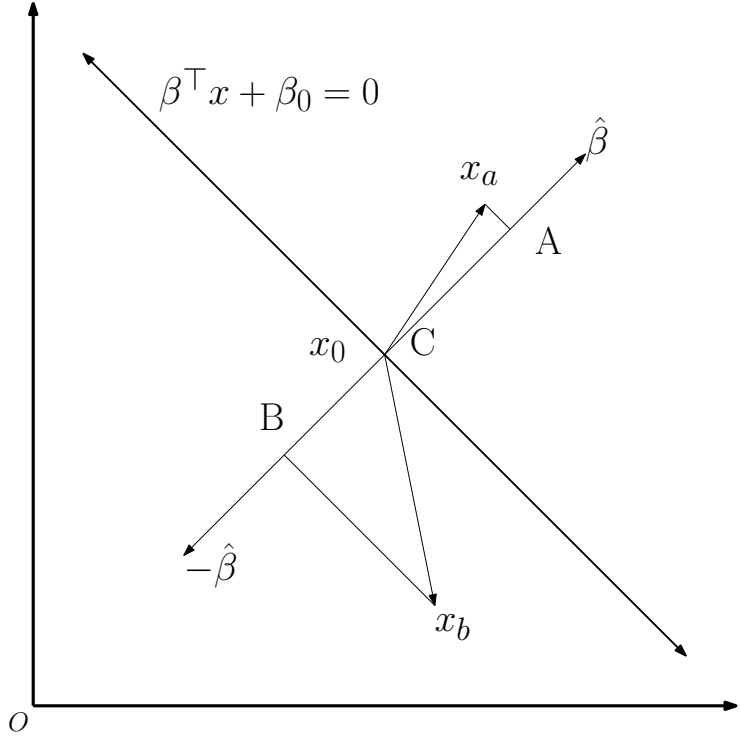}
	\caption{Mathematical formulation of a linear support vector machine, $x_a\in X_a$ (class-$a$) and $x_b\in X_b$ (class-$b$).}
	\label{fig:svm_derivation}
\end{figure}

Assume that Slater's condition holds. Since strong duality holds for \eqref{standard_SOP}, $(x^{*},\lambda^{*},\mu^{*})$ satisfying the KKT conditions \eqref{standard_KKT} is a saddle point of the Lagrangian $\mathcal{L}$. This implies, $x^{*}$ is an optimal solution to primal problem \eqref{standard_SOP} and $\left(\lambda^{*},  \mu^{*}\right)$ is optimal solution to its dual problem \eqref{standard_SOP_dual}, that is
\begin{equation}
(x^\ast,\lambda^\ast,\mu^\ast)=arg\max_{\lambda,\mu}\left(arg\min_{x}\mathcal{L}(x,\lambda,\mu)\right).
\end{equation}
This gives us the following saddle-point dynamics.
\begin{eqnarray}\label{primal-dual-dyn}
-\tau_{x}\dot{x}&=&\left(\nabla_{x} f(x)+\sum_{i=1}^{m}\lambda_{i}\nabla_{x} h_{i}(x)+\sum_{i=1}^{p}\mu_{i}\nabla_{x} g_{i}(x)\right)\nonumber\\
\tau_{\lambda}\dot{\lambda}&=& h(x)\nonumber\\
\tau_{\mu}\dot{\mu}&=& \left(g(x)\right)^{+}_{\mu}
\end{eqnarray}
$\tau_{x}, \tau_{\lambda}\deff\text{diag}\{\tau_{\lambda_{i}},\ldots,\tau_{\lambda_{m}}\}, \tau_{\mu}\deff\text{diag}\{\tau_{\mu_{i}},\ldots,\tau_{\mu_{p}}\}$ are positive definite matrices and $\left(g(x)\right)^{+}_{\mu_i}$ is given by
\beq\label{IED1}
\left(g(x)\right)^{+}_{\mu}=\begin{cases}
	g_{i}(x) \;\;\;\;\;\;\;\;\;\;\;\;\;\;\text{if}\; \mu_{i}>0 \;\; \forall i \in \{1,\hdots, p\}, \label{PriD}\\
	\text{max}(0,g_{i}(x))\;\; \text{if} \;\mu_{i}=0.
\end{cases}
\eeq
\begin{remark}The discrete time primal-dual gradient descent equations of convex optimization problem \eqref{standard_SOP} are 
	\begin{eqnarray*}
		x(t_{k+1})&=&x(t_k)-\eta_{x}\nabla_x\mathcal{L}(x,\lambda,\mu)\\
		\lambda(t_{k+1})&=&\lambda(t_k)+\eta_{\lambda}\nabla_{\lambda}\mathcal{L}(x,\lambda,\mu),\;\;\; k\in \mathbb{Z}^+\\
		\mu(t_{k+1})&=&\mu(t_k)+\eta_{\mu}\left(\nabla_{\lambda}\mathcal{L}(x,\lambda,\mu)\right)^{+}_{\mu},\;\;\; k\in \mathbb{Z}^+.
	\end{eqnarray*}
	where $\eta_x>0$, $\eta_{\lambda}>0$ and $\eta_{\mu}>0$ represents the step size. Further these are equivalent to the continuous time equations \eqref{primal-dual-dyn}, if the step sizes are chosen as $\eta_x=\Delta T\tau_x^{-1}$, $\eta_{\lambda}=\Delta T\tau_{\lambda}^{-1}$ and $\eta_{\mu}=\Delta T\tau_{\mu}^{-1}$, where $\Delta T=t_{k+1}-t_k$.
\end{remark} 

\noindent {\em Equality constrained optimization problem:}
 Consider the following dynamics
\begin{equation}\label{maindyn}
\begin{split}
-\tau_{x}\dot{x}&=\nabla_{x} f(x)+\sum_{i=1}^{m}\lambda_{i}\nabla_{x} h_{i}(x)+u\\
\vspace{-100mm}
\tau_{\lambda_{i}}\dot{\lambda}_{i}&= h_{i}(x),\;\;~\hspace{2cm}
y =-x,
\end{split}
\end{equation}
where $u, y \in \mathbb{R}^n$. Note that, the unforced system of equations, obtained by setting $u=0$ in \eqref{maindyn}, represent primal-dual dynamics corresponding to convex optimization problem \eqref{standard_SOP} with only equality constraints.
\begin{proposition}\label{prop::eq_const}
	Let $\bar{z}=(\bar{x}, \bar{\lambda})$ represent the unforced equilibrium of \eqref{maindyn}. Assume $h(x)$ is  convex and  $f(x)$ strictly convex. Then the system of equations \eqref{maindyn} are passive with port variables $(\dot{u},\dot{y})$ \cite{CSS_letters,ICCvdotidot}.  Further every solution of the unforced version ($u=0$) of  \eqref{maindyn} asymptotically converges to $\bar{z}$.
\end{proposition}
\noindent{\em Inequality constraint: }
We now define the inequality constraint $g_i(\tilde{u})\leq 0$  as the following hybrid dynamics
\begin{equation}\label{IED}
\tau_{\mu_i}\dot \mu_i=(g_i(\tilde u))^+_{\mu_i}
\end{equation}
where  $\tilde{u}\in \mathbb{R}^n$ and $i\in \{1\cdots p\}$. 
This is introduced in \cite{kose1956solutions}, where the authors construct a dynamical system which converges to the stationary solution of saddle value problems. These equations are proposed in such a way that, if the initial condition of $\mu(t)$ is non-negative, then the trajectories $\mu(t)$ always stay inside positive orthant $\mathbb{R}^+$. 
Note that the discontinuity in the above equations \eqref{IED1} occurs when $g_i(\tilde{u})<0$ and $\mu_i=0$, the value of $(g_i(\tilde{u}))^+_{\mu_i}$ switches from $g_i(\tilde{u})$ to $0$. This ensures that the $\mu_i$'s does not go below zero. To make this more visible, we redefine these equations equivalently as follows; 
Let $\mathcal{P}$ represent the power set of $\{1\cdots p\}$, then we define the function $\sigma:[0,\;\infty)\rightarrow \mathcal{P}$ as follows
\beq\label{sigma_map} \sigma(t)=\{i \mid \; \mu_i(t)=0\;\;\text{and }\;g_i(\tilde{u})\leq0\, \forall i \in \{1, ..., p\}\}.\eeq
 With $\sigma(t)$ representing the switching signal, equation \eqref{IED} now takes the form of a switched system 
\beq \label{active_const_def}
\tau_{\mu_i}\dot \mu_i=g_i(\tilde u,\sigma)&=&\left\{\begin{matrix}
	g_i(\tilde u) ;&\; i\notin \sigma(t)\\0 ;&\;i\in \sigma(t)
\end{matrix} \right. 
\eeq
The overall dynamics of the $p$ inequality constraints $g_i(\tilde{u})\leq 0$ $\forall i\in \{1\cdots p\}$ can be written in a compact form as:
\begin{equation}\label{IEdyn}
\tau_{\mu}\dot{\mu}=g(\tilde u,\sigma)
\end{equation}
where $\mu_i$ and $g_i(\tilde{u},\sigma)$ are $i^{th}$ components of  $\mu$ and $g(\tilde{u},\sigma)$ respectively.
%
%
%
Consider the following storage function(s) 
\beq\label{storage_fun_ineq_const}
S_{\sigma_q}(\mu)&=&\dfrac{1}{2}\sum_{i\notin\sigma_q}^{}\dot \mu_i^2\tau_{\mu_i}\;\;\; \forall \sigma_q \in \mathcal{P}
\eeq

\begin{proposition}\cite{CSS_letters,HybridPassive} \label{prop:ineq_passivity}
	The switched system \eqref{IEdyn} is  passive with switched storage functions $S_{\sigma_q}$ (defined one for each switching state $\sigma_q\in\mathcal{P}$ ), input port $u_s=\dot{\tilde{u}}$ and  output port $y_s=\dot{\tilde{y}}$ where $\tilde{y}=\sum_{\forall i} \mu_i\nabla_{\tilde u}g_i(\tilde{u})$. That is, for each $\sigma_p\in \mathcal{P}$ with the property that for every pair of switching times $(t_i,t_j)$, $i<j$ such that $\sigma(t_i)=\sigma(t_j)=\sigma_p\in \mathcal{P}$ and $\sigma(t_k)\neq \sigma_p$ for $t_i<t_k<t_j$, we have
	\begin{eqnarray}\label{Hybrid_passivity}
	S_{\sigma_p}(\mu(t_j))-S_{\sigma_p}(\mu(t_i))\leq \int_{t_i}^{t_j}u_s^\top y_sdt.
	\end{eqnarray}
\end{proposition}
\begin{proposition}\label{prop::ass_stab_ineq}
	The equilibrium set $\Omega_e$ defined by constant control input $\tilde{u}=\tilde{u}^\ast$ of \eqref{IED}
	\beqn\label{IE_equilibrium}
	\Omega_e=\left\{(\bar{\mu},\tilde{u}^\ast)\left|g_{i}(\tilde{u}^{*})\leq 0,\;\; \bar{\mu}_{i}g_{i}(\tilde{u}^{*})=0 \hspace{0.2cm} \forall i \in \{1,\hdots, p\}\right.\right\}
	\eeqn
	is asymptotically stable.
\end{proposition}

\begin{figure*}
	\centering
	\includegraphics[width=\textwidth]{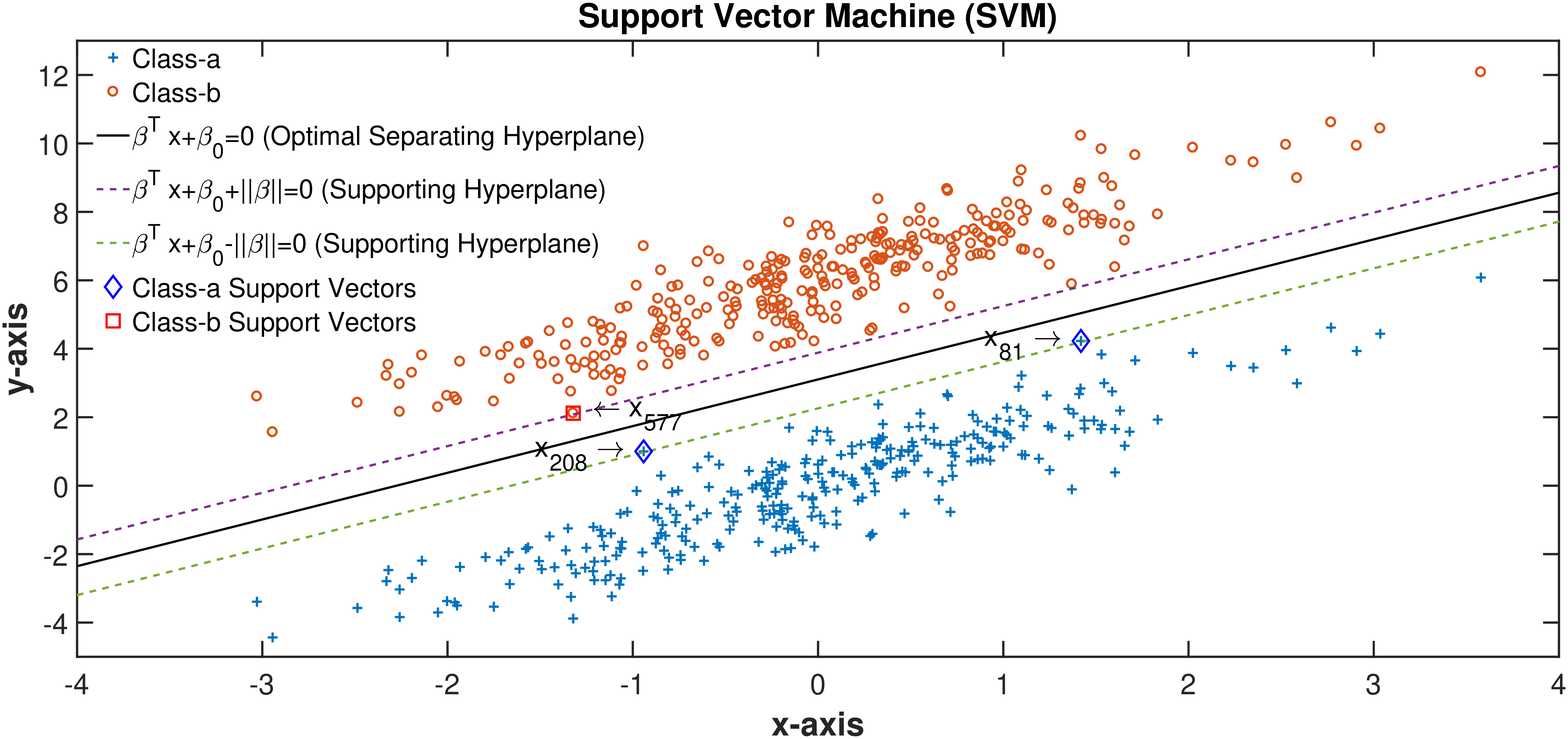}
	\caption{Classification using Support Vector Machine}
	\label{fig::svm3}
\end{figure*}

{\em The overall optimization problem:}
\label{chap::4::sec::3}
We now define a power conserving interconnection between passive systems associated with optimization problem with an equality constraint \eqref{maindyn} and an inequality constraint \eqref{IED} (see Fig. \ref{fig:my_label}).
\begin{proposition} \label{interconnectpassive}\cite{CSS_letters}
	Consider the interconnection of passive systems \eqref{maindyn} and  \eqref{IED}, via the following interconnection constraints $ u=\tilde{y}+v~~ \text{and} ~~\tilde{u}=-y+\tilde{v}, ~ v\in \mathbb{R}^p, ~ \tilde v\in \mathbb{R}^n $. 
	For $\tilde v=0$, the interconnected system is then passive with port variables $\dot{v}$, $-\dot{x}$ (see Fig. \ref{fig:my_labela}). 
	Moreover for $v=0$ and $\tilde v=0$ the interconnected system represents the primal-dual gradient dynamics of the optimization problem \eqref{standard_SOP}
	and the trajectories converge asymptotically to the optimal solution of \eqref{standard_SOP}.
\end{proposition}
\noindent         In the next section, we demonstrate the continuous-time primal-dual algorithm, on the convex optimization formulation of Support Vector Machines (SVM) technique \cite{cortes1995support}. 

\section{Linear SVM as an application}

Support Vector Machines \cite{cortes1995support} are a class of supervised machine learning algorithms which are commonly used for data classification. 
In this methodology, each data item is a point in $n$-dimensional space that is mapped to a category (or a class). 
Here the aim is to find an optimal separating hyperplane (OSH) which separates both the classes and maximizes the distance to the closest point from either class (as shown in Figure \ref{fig::svm_ex}). These closest points are usually called support vectors (SV). The lines passing through support vectors and parallel to the optimal separating hyperplane are called supporting hyperplanes (SH).

{\em Problem formulation:} Consider two linearly separable classes, where each class (say class-$a$, class-$b$) contains a set of $N$ unique data points in $\mathbb{R}^2$. Let $X_a$ and $X_b$ denote the set of points in class-$a$ and class-$b$ respectively. In this methodology we find a hyperplane that separates the classes while maximizing the distance to the closest point from either class. Let $L$ be an affine set that characterizes such a hyperplane, defined as follows
\beq
L=\left\{x\in \mathbb{R}^2|x^\top \beta+\beta_0=0 \right\}
\eeq
where $\beta=(\beta_1,\beta_2)\in \mathbb{R}^2$ and $\beta_0\in \mathbb{R}$.  Define the map $l:\mathbb{R}^2\rightarrow \mathbb{R}$ by $l(x)=x^\top \beta +\beta_0$. Note the following, for any $x_0\in L$, $l(x_0)=0$ $\implies$ $x_0^\top \beta =-\beta_0$. This implies $l(x)$ can be rewritten as $l(x)=\beta^\top (x-x_0)$, which further implies the unit vector $\hat{\beta}=\dfrac{\beta}{||\beta||}$ is orthogonal to the line defined by the set $L$, that is, $x^\top\beta+\beta_0=0\iff (x-x_0)^\top \beta=0$.
%

The distance between the point $x_a\in X_a$ and line $L$ is $|AC|=(x_a-x_0)^\top \hat{\beta}$ (see Fig. \ref{fig:svm_derivation}). Similarly the distance between the point $x_b\in X_b$ and line $L$ is $|BC|=(x_b-x_0)^\top( -\hat{\beta})$. We want to find an optimal separating hyperplane that is at least $M$ units away from all the points. This implies
\beq\label{svm_constraints}
\begin{matrix}
	\forall x_a\in X_a, &  	(x_a-x_0)^\top \hat{\beta}     &\geq& M,\\
	\forall x_b\in X_b, & -(x_b-x_0)^\top \hat{\beta}&\geq& M.
\end{matrix}
\eeq
Define $X\deff X_a\cup X_b$, and  $Y\deff Y_a \cup Y_b$ where $Y_a=\underbrace{\{1,\ldots,  1\}}_{\text{n\;\; times}}$ and $Y_b=\underbrace{\{-1,\ldots,  -1\}}_{\text{n\;\; times}}$. 
The inequality constraints \eqref{svm_constraints} can be rewritten as
\beq 
\dfrac{1}{||\beta||}y_i(\beta^\top x_i+\beta_0)\geq M
\eeq
where $y_i=1$ if $x_i\in X_a$ (class-$a$), $y_i=-1$ if $x_i\in X_b$ (class-$b$). Finally, finding the optimal separating hyperplane can be proposed as the following optimization problem,
\begin{equation}\label{convex_form1}
\begin{aligned}
& \underset{\beta, \beta_0}{\text{maximize}}
& & M\\
& \text{subject to}
& & \dfrac{1}{||\beta||}y_i(\beta^\top x_i+\beta_0)\geq M,~\;\; \forall x_i\in X,~ y_i\in Y.
\end{aligned}
\end{equation} 
Since $M$ is arbitrary, choosing $M=\dfrac{2}{||\beta||}$ converts \eqref{convex_form1} into a convex optimization problem
\begin{equation}\label{convex_form2}
\begin{aligned}
& \underset{\beta, \beta_0}{\text{minimize}}
& & \dfrac{1}{2}||\beta||\\
& \text{subject to}
& & y_i(\beta^\top x_i+\beta_0)\geq 1,~\;\; \forall x_i\in X,~ y_i\in Y.
\end{aligned}
\end{equation} 
In order to use the primal-dual gradient method proposed in Section \ref{chap::4::sec::3}, we need the cost function to be twice differentiable. But, the cost function $\frac{1}{2}||\beta||\notin C^2$. 
The optimal solution $(\beta^\ast,\beta_0^\ast)$ of \eqref{convex_form2}, is further equivalent to the optimal solution of
\begin{equation}\label{convex_form3}
\begin{aligned}
& \underset{\beta, \beta_0}{\text{minimize}}
& & \dfrac{1}{2}||\beta||^2\\
& \text{subject to}
& & y_i(\beta^\top x_i+\beta_0)\geq 1,~\;\; \forall x_i\in X,~ y_i\in Y.
\end{aligned}
\end{equation}
We now use this convex optimization formulation for support vector machines, and derive its primal-dual gradient dynamics.\\
{\em Continuous time primal-dual gradient dynamics}: Comparing with the convex optimization formulation given in \eqref{standard_SOP}, the cost function is $f(\beta)=\dfrac{1}{2}||\beta||^2$ and inequality constraints are $g_i(\beta,\beta_0)=1-y_i(\beta^\top x_i+\beta_0)$, $i\in \{1,\cdots, 2N\}$. The Lagrangian can be written as
\beq
L(\beta,\mu)=\dfrac{1}{2}||\beta||^2+\sum_{i=1}^{2N}g_i(\beta,\beta_0)\mu_i
\eeq
where $\mu=(\mu_1,\cdots,\mu_{2N})$ denotes the Lagrange variable corresponding to the inequality constraints $g=(g_1,\cdots, g_{2N})$. The primal dual gradient laws given in \eqref{primal-dual-dyn} for the convex optimization problem \eqref{convex_form3} are
\beqn
\begin{matrix}
	-\tau_{\beta}\dot{\beta}&=&\dfrac{\partial L}{\partial \beta}\\
	-\tau_{\beta_0}\dot{\beta}_0&=&\dfrac{\partial L}{\partial \beta_0}\\
	\tau_{\mu_i}\dot{\mu}_i&=&\left(\dfrac{\partial L}{\partial \mu_i}\right)^+_{\mu_i}\;\; \forall i \in \{1,\hdots, 2N\}
\end{matrix}
\eeqn
equivalently ,
\beq\label{primal-dual_svm}
-\tau_{\beta}\dot{\beta}&=&\beta-\sum_{i=1}^{2N} \mu_iy_ix_i\nonumber\\
-\tau_{\beta_0}\dot{\beta}_0&=&-\sum_{i=1}^{2N} \mu_iy_i\\
\tau_{\mu_i}\dot{\mu}_i&=&(g_i(\beta,\beta_0))^+_{\mu_i}\;\; \forall i \in \{1,\hdots, 2N\}
\nonumber
\eeq 
Note that the equilibrium point of the above dynamical system \eqref{primal-dual_svm} represents the KKT conditions of the optimization problem \eqref{convex_form3}. The first two equations represents the KKT conditions with respect to primal variables and third equation represents the complimentary conditions for the dual variables. This implies, finding the solution of the equilibrium point is equivalent to solving the KKT conditions, which is not a trivial task in many cases. Hence, the equilibria of the above dynamical system is not explicitly known but is implicitly characterized by the optimization problem. Instead of solving for these equilibrium points manually, one can run the dynamical system and use its steady state behavior (points). But to quantify it mathematically, we first have to prove that the dynamical system is globally asymptotically stable at that equilibrium point. To do that we leverage the propositions presented in the previous sections. We now have the following result.
\begin{proposition}
	The primal-dual dynamics \eqref{primal-dual_svm} converges asymptotically to the optimal solution of \eqref{convex_form3}.
\end{proposition}
\begin{proof}
	Since the optimization problem \eqref{convex_form3} has a strictly convex cost function and convex inequality constraints, the result follows from Propositions \ref{prop::eq_const} - \ref{interconnectpassive}.
\end{proof}

\subsection{Simulation Results}
A simulation study is conducted by generating two sets of linearly separable classes having 300 points each, using Normal distribution (see Table \ref{tab:data_gen} for distribution parameters). Figure \ref{fig::svm2a} present the evolution of $\beta,\;\beta_0$. 
\begin{table}[h!]
	\centering
	\caption{Distribution parameters}
	\label{tab:data_gen}
	\begin{tabular}{|l|l|l|l|}
		\hline
		\textbf{}        & \textbf{mean}                     & \textbf{Variance}                          & \textbf{No. of data points} \\ \hline
		\textbf{Class-a} & $\begin{bmatrix}0&0\end{bmatrix}$ & $\begin{bmatrix}1&1.5\\1.5&3\end{bmatrix}$ & 300                       \\ \hline
		\textbf{Class-b} & $\begin{bmatrix}0&6\end{bmatrix}$ & $\begin{bmatrix}1&1.5\\1.5&3\end{bmatrix}$ & 300                         \\ \hline
	\end{tabular}
\end{table}
At equilibrium, the primal-dual dynamics in equation \eqref{primal-dual_svm} results in
$$\beta^{\ast}=\sum_{i=1}^{2N} \mu_i^{\ast}y_ix_i.$$
\begin{remark} The results depicted in Figure \ref{fig::svm2} show that the value the Lagrange variables, except ($ \mu_{81},\; \mu_{208},\;\;\mu_{577}$) are identically equal to zero at equilibrium. 
 Moreover, the data points $x_{81}$, $x_{208}$ and $x_{577}$ corresponding to these non-zero Lagrange variables are called support vectors, can be seen in Fig. \ref{fig::svm3}.  The lines passing through these point and parallel to the separating hyperplane are called supporting hyperplanes. 
\end{remark}
\noindent Hence 
$$\beta^\ast=\mu_{81}^\ast x_{81}+\mu_{208}^\ast x_{208}-\mu_{577}^\ast x_{577}$$
where the data points ($ x_{81},\; x_{208},\;\;x_{577}$) corresponding to these non zero Lagrange variables are support vectors. This implies that the support vectors completely determines the optimal separating hyperplane $\beta^\top x+\beta_0=0$ that separates class-$a$ and class-$b$ (see Fig. \ref{fig::svm1}). However, note that one needs to solve the optimization problem, to find these support vectors.

\begin{remark}
	Figure \ref{fig::svm1} shows that, whenever, an inequality constraint becomes feasible  (i.e. $g_i(\beta,\beta_0)\leq  0$ ) and its corresponding Lagrange variable $\mu_i$ converges to zero, then the closed loop storage function \eqref{storage_fun_ineq_const} switches to a new storage function that is strictly less than the current one, causing a discontinuity. 
	%
	This is coherent with the Proposition \ref{prop:ineq_passivity}, where passivity property is defined with `multiple storage functions'.
\end{remark}
\begin{figure}[t]
	\includegraphics[trim={0cm 0cm 0 0cm},width=9cm,height=6cm]{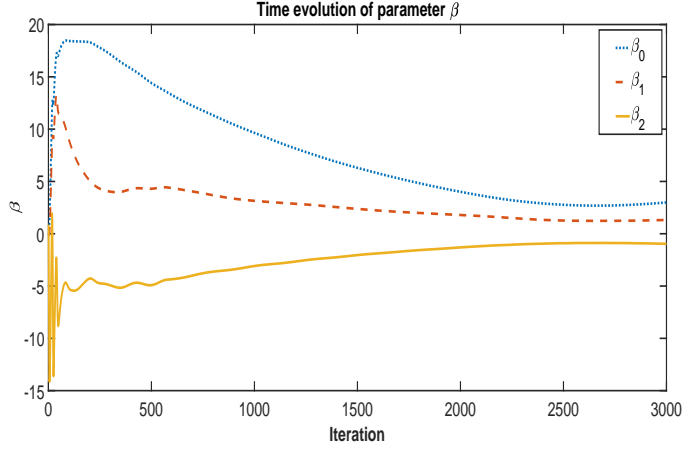}
	\caption{Time evolution of $\beta$ and $\beta_0$ }
	\label{fig::svm2a}
\end{figure}
\begin{figure}[t]
	\includegraphics[trim={7cm 0cm 0cm 0cm}, clip, scale=0.25]{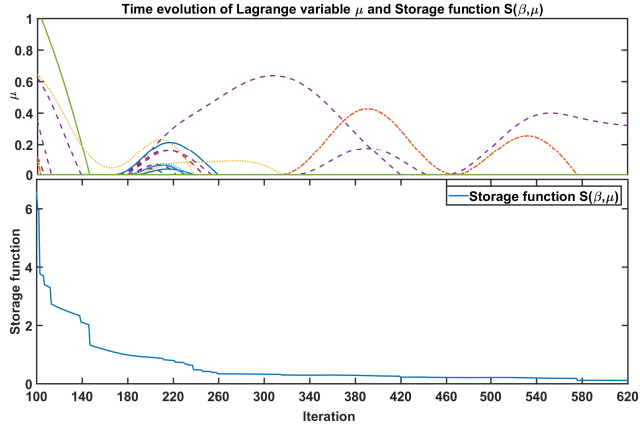}
	\caption{Time evolution of closed-loop storage function. }
	\label{fig::svm1}
\end{figure}
\begin{figure}[t]
	\includegraphics[trim={5cm 0cm 0 0cm},clip, width=9.5cm,height=6cm]{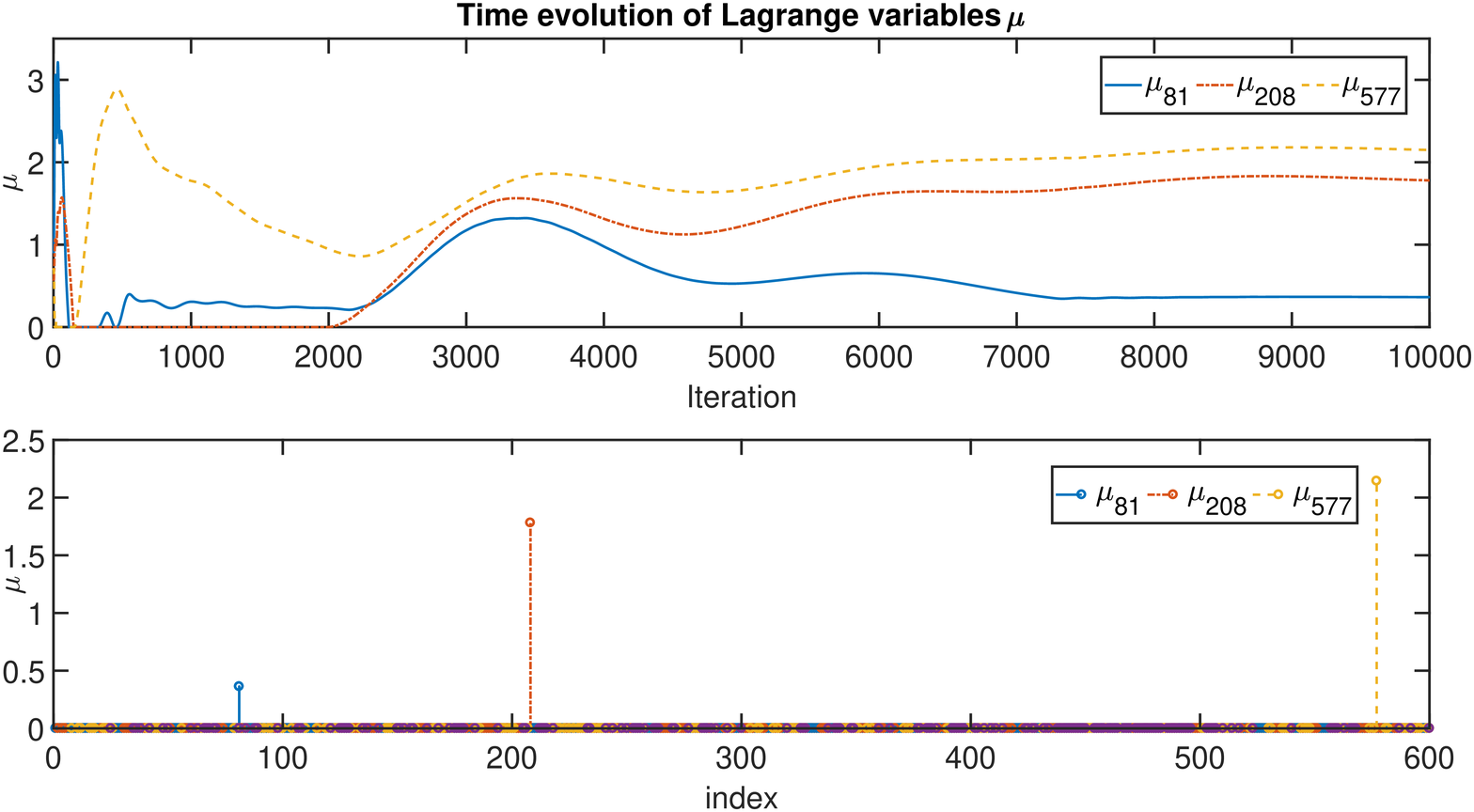}
	\caption{Time evolution of Lagrange variables $\mu_i$, $i\in \{1\cdots 600\}$. }
	\label{fig::svm2}
\end{figure}


\section{Future work}
In Section III, the primal-dual algorithm is treated as interconnected passive systems, (i) convex optimization problem with only equality constraint, (ii) a state dependent switching system for inequality constraint. Recall that in Proposition 4, we interconnected these systems using
\beq\label{interconnection_ch6}
\begin{bmatrix}
	u\\\tilde{u}
\end{bmatrix}&=& \begin{bmatrix}
0 & 1\\-1 &0
\end{bmatrix}\begin{bmatrix}
y\\ \tilde{y}
\end{bmatrix}+\begin{bmatrix}
v\\\tilde{v}
\end{bmatrix}
\eeq
where $v$ and $\tilde{v}$ are considered as new input port-variables of the interconnected system. We can use these new port-variables to analyze and improve the primal-dual gradient laws. The following are some of the important ideas that can be  leveraged for future work.

\noindent {\em Robustness}: To analyze uncertainties in parameters or disturbances such as the numerical error accumulated in the primal and dual variables, one can rewrite interconnection as
	   \begin{equation}
	   u=\tilde{y}+\Delta \tilde{y}~~ \text{and} ~~\tilde{u}=x+\Delta x
	  \end{equation}
	   where $\Delta x$ and $\Delta \tilde{y}$ denotes the numerical error in $x$ (primal variable) and $\tilde{y}$ (a function of dual variable) respectively. These can be treated as external disturbances  creeping in through the interconnected port variables. 
We can provide robustness analysis quantitatively (on sensitivity of the algorithm due to numerical errors), using input/output dissipative properties \cite{l2gain} of these systems.

\noindent {\em Stochastic gradient descent:} 
In SVM simulation we have seen that there are 600 inequality constraints (each corresponds to a data-point). Usually, real world examples may contain many more data-points. Each data-points gives rise to an inequality constraint, and further leads to a gradient-law. In situations involving large data, it is computationally ineffective to run gradient-descent algorithm using all the data-points. In general this obstacle is circumvented using a variation in gradient descent method called {\em stochastic gradient descent}. Can we propose a passivity based convergence analysis for stochastic gradient descent?

\noindent {\em Control synthesis}: Using these new port variables one can interconnect the primal-dual dynamics to a plant, such that the closed-loop system is again a passive dynamical system \cite{stegink2017unifying}. One can also explore the idea of Barrier functions \cite{boyd2004convex} to derive a bounded controller. Gradient methods are inherently distributed computing methods. Hence the controllers derived from these may inherit this property. Moreover, this framework enables us to solve control problems that whose operating points are characterized by an optimization problems.
\bibliographystyle{IEEEtran}
\bibliography{refs}

\end{document}